\documentclass[journal,twocolumn,twoside,10pt]{IEEEtran}
\usepackage{graphicx}
\usepackage{bm}
\usepackage[cmex10]{amsmath}
\usepackage{amssymb}
\usepackage{acronym}
\usepackage{setspace}
\usepackage{amsthm}
\usepackage{footnote}
\usepackage{cite}

\newtheorem{theorem}{Theorem}
\newtheorem{lemma}{Lemma}

\newtheorem{corollary}{Corollary}

\graphicspath{{figure/}}

\ifCLASSINFOpdf

\else

\fi
\begin{document}

\title{Locally Orthogonal Training Design for Cloud-RANs Based on Graph Coloring}
\author{Jianwen~Zhang,
		Xiaojun~Yuan,~\IEEEmembership{Senior Member,~IEEE},
		and Ying Jun (Angela) Zhang,~\IEEEmembership{Senior Member,~IEEE}
\thanks{J. Zhang and X. Yuan are with the School of Information Science and Technology, ShanghaitTech University, Shanghai, China, email: \{zhangjw1, yuanxj\}@shanghaitech.edu.cn.}
\thanks{Y. Zhang is with the the Department of Information Engineering, The Chinese University of Hong Kong, Shatin, New Territories, Hong Kong. e-mail: yjzhang@ie.cuhk.edu.hk.}
\thanks{Part of this work has been submitted to IEEE Global Communications Conference (GlobeCom2016).}
		}



\maketitle

\begin{abstract}
	We consider training-based channel estimation for a cloud radio access network (CRAN), in which a large amount of remote radio heads (RRHs) and users are randomly scattered over the service area. In this model, assigning orthogonal training sequences to all users will incur a substantial overhead to the overall network, and is even impossible when the number of users is large. Therefore, in this paper, we introduce the notion of \emph{local orthogonality}, under which the training sequence of a user is orthogonal to those of the other users in its neighborhood. We model the design of locally orthogonal training sequences as a graph coloring problem. Then, based on the theory of random geometric graph, we show that the minimum training length scales in the order of $\ln K$, where $K$ is the number of users covered by a CRAN. This indicates that the proposed training design yields a scalable solution to sustain the need of large-scale cooperation in CRANs. Numerical results show that the proposed scheme outperforms other reference schemes.
\end{abstract}

\IEEEpeerreviewmaketitle

\section{Introduction}
	\IEEEPARstart{C}{loud} radio access network (CRAN), which exhibits significant improvement on spectrum efficiency, is one of the enabling technologies for future 5G wireless communications \cite{5Gkey}. The main idea of CRAN is to split a base station into a remote radio head (RRH) for radio frequency signaling and a baseband unit (BBU) for baseband signal processing. BBUs, centralized in a BBU pool \cite{cran11, Checko2015}, are connected to RRHs via high-capacity fronthaul links, as illustrated in Fig. \ref{Fig:CRAN_Structure}. The CRAN technology significantly improves the system throughput via ultra-dense RRH deployment and centralized control \cite{cran11}. However, a CRAN involves cooperation among hundreds and even thousands of RRHs and users. Such a large-scale cooperation imposes a stringent requirement on channel state information (CSI). The acquisition, distribution, and storage of CSI may incur a considerable overhead to the network. As such, it is highly desirable to design an efficient channel training scheme for CRANs.
	\begin{figure}[!t]
		\centering
		\includegraphics[width = 0.43\textwidth]{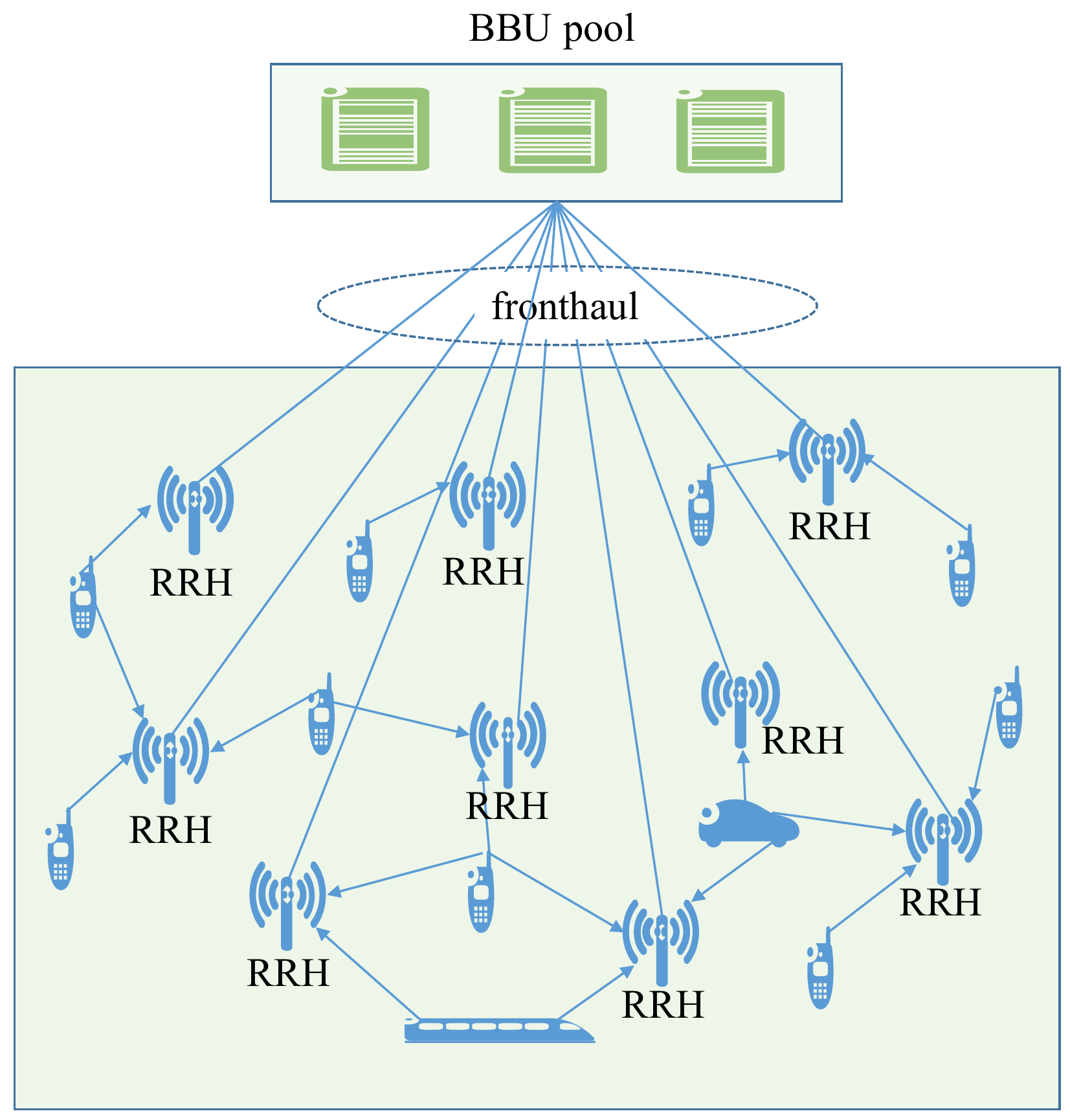}
		\caption{CRAN architecture.}\label{Fig:CRAN_Structure}
	\end{figure}
	
	Time-multiplexed training design has been studied in various communication models \cite{Coldrey2007,Xiaojun2016,Marzetta2010,Fernades2013,Mugen15, Jma2014, Upadhya2016, Ngo2012, Muller2013, Huh12, Zhu15, Marzetta15, Yin13, Caire2016, JunZhang2016, Ahmadi2015, Jose2011}. In particular, the authors in \cite{Coldrey2007} derived the optimal time-multiplexed training design for conventional point-to-point multiple-input multiple-output (MIMO) systems, where the transmit antennas are co-located and so are the receive antennas. In \cite{Xiaojun2016}, the authors studied the optimal time-multiplexed training design for a multiuser MIMO channel where users are scattered and suffer from the near-far effect. Orthogonal training sequences were shown to be optimal or nearly optimal in these scenarios \cite{Coldrey2007,Xiaojun2016}. However, orthogonal training design is very inefficient when applied to CRAN, for that a CRAN system usually covers a large number of users and RRHs. Allocating orthogonal training sequences to users inevitably leads to an unaffordable overhead to the system.
	
	In this paper, we investigate the design of training sequences for time-multiplexed channel training in CRANs. A crucial observation in a CRAN is that both RRHs and users are scattered over a large area. As such, due to severe propagation attenuation of electromagnetic waves, the interference from far-off users can be largely ignored when processing the received signal of an RRH. Therefore, rather than global orthogonality, we introduce the notion of \emph{local orthogonality}, in which the training sequences of the neighboring users with distance no greater than a certain threshold (denoted by $r$) are required to be orthogonal to each other. The training design problem is then formulated as to find the minimum training length that preserves local orthogonality. This problem can be recast as a vertex-coloring problem, based on which the existing vertex-coloring algorithms \cite{Brelaz1979,Malaguti2008} are applicable. Further, we analyze the minimum training length as a function of the network size. Based on the theory of random geometric graph, we show that the training length is $O(\ln K)$ almost surely, where $K$ is the number of users. This guarantees a scalable training-based CRAN design, i.e., the proposed training design can be applied to a large-size CRAN system satisfying local orthogonality at the cost of a moderate training length.
	
	In the proposed scheme, the neighborhood of an RRH is defined as the area centered around it with distance below the threshold $r$. For a large $r$, the neighborhood of an RRH is large and more multi-user interference from neighboring users can be eliminated in channel estimation. Then, local orthogonality achieves a channel-estimation accuracy close to that of global orthogonality. However, a larger neighborhood area implies more channel coefficients to be estimated, thereby incurring a larger overhead to the system. Therefore, there is a balance to strike between the accuracy and the overhead of channel estimation. In this paper, we study this tradeoff for throughput maximization. We show that, with local orthogonality, the optimal $r$ for throughput maximization can be numerically determined. 
	
\subsection{Related Work}
	In the considered training-based CRAN, orthogonal training sequences do not interfere with each other; the interference only comes from far-off users with non-orthogonal training sequences. This is similar to the problem of \emph{pilot contamination} in multi-cell massive MIMO systems \cite{Marzetta2010}, where the orthogonal training sequences used in each cell are reused among cells.
	
	There are several existing approaches to tackle the pilot contamination problem \cite{Marzetta2010, Fernades2013, Mugen15, Jma2014, Upadhya2016, Ngo2012, Muller2013, Huh12, Zhu15, Marzetta15, Yin13, Caire2016, JunZhang2016, Ahmadi2015, Jose2011}. For example, data-aided channel estimation with superimposed training design was proposed in \cite{Jma2014,Mugen15,Upadhya2016} to reduce the system overhead spent on channel estimation and to suppress pilot contamination. In \cite{Ngo2012,Muller2013}, the authors proposed blind channel estimation based on singular value decomposition (SVD). However, both superimposed training and blind channel estimation involve high computational complexity in implementation, especially when applied to a CRAN with a large network size.
	
	Time-multiplexed training design has also been considered to address the problem of pilot contamination \cite{Huh12,Zhu15,Marzetta15,Yin13}. The key issue is the design of the training-sequence reuse pattern among cells. In \cite{Huh12} and \cite{Zhu15}, users within each cell are classified into cell-edge and cell-center users. The pilots for cell-center users are reused across all cells, while orthogonal pilot sub-sets are assigned to the cell-edge users. In \cite{Marzetta15}, the cells surrounding the home cell by one or more rings are assigned orthogonal pilot sets. In \cite{Yin13}, training sequences are assigned to users based on the angle of arrival. However, the notion of cell is no longer adopted in CRAN, as RRHs in a CRAN are fully connected to enables full-scale cooperation. Therefore, the cell-based techniques in \cite{Huh12,Zhu15,Marzetta15,Yin13} are not applicable to CRAN. 

	It is also worth mentioning that training-based CRAN has been previously studied in the literature \cite{Caire2016,JunZhang2016}. In \cite{Caire2016}, the authors proposed a coded pilot design where RRHs can be turned on or off to avoid pilot collisions, which may degrade the system performance. In \cite{JunZhang2016}, in each transmission block, only a portion of users is allowed to transmit pilots for channel training, and the channels of the other users are not updated. This scheme can only accommodate a relatively small number of users to avoid an unaffordable training overhead. Therefore, training design for CRAN deserves further endeavor, which is the main focus of this work.

\subsection{Organization}
	The remainder of this paper is organized as follows. In Section II, we describe the system model. The definition of local orthogonality and the problem formulation are presented in Section III. In Section IV, we introduce our training sequence design algorithm. In Section V, we characterize the minimum training length as a function of the size of CRAN. The practical design and numerical results are given in Section VI. Section VII concludes this paper.

\subsection{Notation}
	Regular letters, lowercase bold letters, and capital bold letters represent scalars, vectors, and matrices, respectively. Throughout this paper, the vectors are row vectors. $\mathbb{R}$ and $\mathbb{C}$ represent the real field and the complex field, respectively; the superscripts $^\text{H}$, $^\text{T}$, and $^\text{-1}$ represent the conjugate transpose, the transpose, and the inverse, respectively; $|\cdot|$, $\|\cdot \|_2$, $\|\cdot \|_\infty$, and $\det(\cdot)$ represent the absolute value, the $\ell_2$-norm, the $\ell_\infty$-norm, and the determinant, respectively; $\mathbf{a} \perp \mathbf{b}$ means that vector $\mathbf{a}$ is orthogonal to vector $\mathbf{b}$; $\text{diag}\{\mathbf{a}\}$ represents the diagonal matrix with the diagonal specified by $\mathbf{a}$; $\rightarrow$ represents ``tends to'' and $a.s.$ is the abbreviation of almost surely; $\limsup_{K\rightarrow \infty}$ denotes limit superior as $K$ tends to infinity. For any functions $f(x)$ and $g(x)$, $f(x) = O(g(x))$ is equivalent to $\lim_{x\rightarrow \infty} \left|\frac{f(x)}{g(x)}\right| = c$, where $c$ is a constant coefficient.
	
\section{System Model}\label{sec:preliminaries}
	Consider a CRAN consisting of $N$ RRHs and $K$ users randomly distributed over a service area. The RRHs are connected to a BBU pool by the fronthaul. We assume that the capacity of the fronthaul is unlimited, so that the signals received by the RRHs are forwarded to the BBU pool without distortion for centralized signal processing. We also assume that users and RRHs are uniformly distributed over the service area that is a square with side length $r_0$. The result in this paper can be extended to service areas with other shapes. We consider a multiple-access scenario where users simultaneously transmit individual data to RRHs.  The channel is assumed to be block-fading, i.e., the channel remains invariant within the coherence time of duration $T$.  
	
	Suppose that a transmission frame consists of $T$ channel uses. Then, the received signal of RRH $i$ at time $t$ is given by
	\begin{equation}\label{equ:recdS1}
		y_{i,t} = \sum_{k=1}^K h_{i,k} \gamma_{i,k} x_{k,t} + z_{i,t}, i=1,\dots, N, t=1,\dots, T
	\end{equation}
	where $x_{k,t}$ denotes the signal transmitted by user $k$ at time $t$, $z_{i,t}\sim \mathcal{CN}(0, N_0)$ is the white Gaussian noise at RRH $i$, $h_{i,k}$ is the small-scale fading factor from user $k$ to RRH $i$ and is independently drawn from $\mathcal{CN}(0,1)$, and $\gamma_{i,k}$ represents the large-scale fading factor from user $k$ to RRH $i$. In this paper, $\gamma_{i,k}$ is modeled as $\gamma_{i,k} = d_{i,k}^{-\frac{\eta}{2}}$, where $d_{i,k}$ denotes the distance between user $k$ and RRH $i$, and $\eta$ is the path loss exponent. Denote by $\mathbf{y}_i = [y_{i,1},\dots, y_{i,T}] \in \mathbb{C}^{1\times T} $ the received signal at RRH $i$ in a vector form and $\mathbf{x}_k = [x_{k,1},\dots,x_{k,T}] \in \mathbb{C}^{1\times T}$ the corresponding transmitted signal vector of user $k$. The signal model in \eqref{equ:recdS1} can be rewritten as
	\begin{equation}\label{equ:recdS2}
		\mathbf{y}_i = \sum_{k=1}^K h_{i,k} \gamma_{i,k} \mathbf{x}_k + \mathbf{z}_i, i = 1,\dots,N
	\end{equation}
	where $\mathbf{z}_i = [z_{i,1},\dots, z_{i,T}] \in \mathbb{C}^{1\times T}$ is the noise vector at RRH $i$. The power constraint of user $k$ is given by
	\begin{equation}\label{equ:pc}
		\frac{1}{T} \| \mathbf{x}_k \|_2^2 \leq P_0, \quad k = 1,\dots, K
	\end{equation}
	where $P_0$ is the power budget of each user.
	
	We note that if $\gamma_{i,k}=\gamma$ for a certain constant $\gamma$ for all $i$ and $k$, the system in \eqref{equ:recdS2} reduces to a conventional MIMO system where both users and RRHs are co-located. If $\gamma_{i,k}=\gamma_{k}$ for all $i$ and $\gamma_{k}\neq\gamma_{k^\prime}$ for $k\neq k^\prime$, then the system in \eqref{equ:recdS2} reduces to a multiuser system where the RRHs are co-located but the users are separated. In this paper, we consider a general situation that $\gamma_{i,k}\neq\gamma_{i^\prime,k^\prime}$ for $i\neq i^\prime$ or $k\neq k^\prime$, i.e., both users and RRHs are separated from each other.
	
	The large-scale fading coefficients $\{\gamma_{i,k}\}$ only depend on user positions and vary relatively slowly. It is usually much easier to acquire the knowledge of $\{\gamma_{i,k}\}$ than to acquire the small-scale fading coefficients $\{h_{i,k}\}$. Hence, we assume that $\{\gamma_{i,k}\}$ are known at RRHs, while $\{h_{i,k}\}$ need to be estimated based on the received data in a frame-by-frame manner. 
	
	In this paper, we aim to design an efficient transmission scheme to jointly estimate the small-scale fading coefficients $\{h_{i,k}\}$ and detect the signals $\{\mathbf{x}_k \}$. We adopt a two-phase based training scheme consisting of a training phase and a data transmission phase. During the training phase, users transmit training sequences to RRHs for channel estimation. During the data transmission phase, users' data are transmitted and detected at the BBU pool based on the estimated channel. More details follow.
		
\subsection{Training Phase}\label{subsec:training phase}
	Without loss of generality, let $\alpha T$ be the number of channel uses assigned to the training phase, where $\alpha \in (0,1)$ is a time-spliting factor to be optimized. We refer to $\alpha T$ as the training length. From \eqref{equ:recdS2}, the received signal at RRH $i$ for the training phase is given by
	\begin{equation}\label{equ:Training Model}
		\mathbf{y}_{i}^\text{p} = \sum_{k=1}^K h_{i,k} \gamma_{i,k} \mathbf{x}_{k}^\text{p} + \mathbf{z}_{i}^\text{p}, i = 1,\dots,N
	\end{equation}
	where $\mathbf{y}_{i}^\text{p}\in \mathbb{C}^{1\times \alpha T}$ is the received signal at RRH $i$, $\mathbf{x}_{k}^\text{p} =[x_{k,1},\dots, x_{k,\alpha T}] \in \mathbb{C}^{1\times \alpha T}$ is the training sequence transmitted by user $k$, and $\mathbf{z}_{i}^\text{p} \in \mathbb{C}^{1\times \alpha T}$ is the corresponding additive noise. The power constraint for user $k$ in the training phase is given by
	\begin{equation}\label{equ:Training PC}
		\frac{1}{\alpha T} \| \mathbf{x}_{k}^\text{p} \|_2^2 \leq \beta_k P_0, \quad k = 1, 2, \cdots , K
	\end{equation}
	where $\beta_k$ represents the power coefficient of user $k$ during the training phase.
	
\subsection{Data Transmission Phase}\label{subsec:data phase}
	In the data transmission phase, the data of users are transmitted to RRHs and then forwarded to the BBU pool through the fronthaul. The BBU pool performs coherent detection based on the estimated channel obtained in the training phase. From \eqref{equ:recdS2}, the received signal of RRH $i$ in the data transmission phase is written as
	\begin{equation}\label{equ:data model}
		\mathbf{y}_{i}^\text{d} = \sum_{k=1}^K h_{i,k} \gamma_{i,k} \mathbf{x}_{k}^\text{d} + \mathbf{z}_{i}^\text{d},i = 1,\dots,N
	\end{equation}
	where $\mathbf{x}_{k}^\text{d} = [x_{k,\alpha T+1},\dots, x_{k,T}] \in \mathbb{C}^{1\times (1-\alpha)T}$ is the data signal of user $k$, $\mathbf{y}_{i}^\text{d}$ is the corresponding received signal at RRH $i$, and $\mathbf{z}_{i}^\text{d}$ is the corresponding noise. The power constraint of user $k$ in the data transmission phase is given by 
	\begin{equation}
		\frac{1}{(1-\alpha)T} \|\mathbf{x}_{k}^\text{d}\|_2^2 \leq \beta_k^\prime P_0
	\end{equation}
	where the coefficient $\beta_k^\prime = \frac{1-\alpha \beta_k}{1-\alpha}$ satisfies the power constraint in \eqref{equ:pc}. 

\section{Problem Formulation}\label{sec:ProForm}
\subsection{Throughput Optimization}
	The mutual information throughput is a commonly used performance measure for training-based systems \cite{Coldrey2007, Xiaojun2016}. The throughput expression for the proposed training based scheme is derived in Appendix \ref{appendixI}. The system design problem can be formulated as to maximize the throughput over the training sequence $\{ \mathbf{x}_{k}^\text{p} \}$, the training length $\alpha T$, the number of users $K$, and the power allocation coefficients $\{\beta_k\}$ subject to the power constraints in \eqref{equ:pc} and \eqref{equ:Training PC}. Similar problems have been previously studied in the literature. For example, when users are co-located and so are RRHs, the model in \eqref{equ:recdS2} reduces to a conventional point-to-point MIMO system. The optimal training design for throughput maximization was discussed in \cite{Coldrey2007}. Specifically, the optimal strategy is to select a portion of active users while the others keep silent in transmission. The optimal number of active users is equal to $\frac{T}{2}$, and each active user is assigned with an orthogonal training sequence.\footnote{For the MIMO system in \cite{Coldrey2007}, the total transmission power is constrained by a constant invariant to the number of users. The case that the total power linearly scales with $K$ was discussed in \cite{Xiaojun2016}.} Moreover, when only RRHs are co-located, the model in \eqref{equ:recdS1} reduces to a multiuser MIMO system. In this case, users are randomly distributed and suffer from the near-far effect. It was shown in \cite{Xiaojun2016} that the optimal number of active users is in general less than $\frac{T}{2}$, but the orthogonal training design is still near-optimal. A key technique used in \cite{Coldrey2007, Xiaojun2016} is the rotational invariance of the channel distribution when users or RRHs or both are co-located.
	
	This paper is focused on the training design for a general CRAN setting where neither RRHs nor users are co-located. The rotational invariance property of the channel in general does not hold in our setting, and therefore the analysis in \cite{Coldrey2007} and \cite{Xiaojun2016} are no longer applicable. As suggested by the optimal design in \cite{Coldrey2007} and \cite{Xiaojun2016}, it is desirable to design orthogonal training sequences for CRAN. The challenge is that a CRAN usually serves a large number of active users. Thus, assigning every user with an orthogonal training sequence leads to an unaffordable overhead. This inspires us to design the so-called \emph{locally orthogonal} training sequences for CRANs, as detailed below.

\subsection{Local Orthogonality}
	The main advantage of using orthogonal training sequences is that the training signal from one user does not interfere with the training signals from other users. However, the number of available orthogonal sequences is limited by the training length $\alpha T$. This quantity should be kept small so as to reduce the cost of channel estimation. 
	
	A crucial observation in a CRAN is that both RRHs and users are randomly scattered over a large area. Thus, due to the severe propagation attenuation of electromagnetic waves over distance, the interference from far-off users can be largely ignored when processing the received signal of an RRH. This fact inspires the introduction of the channel sparsification approaches in \cite{VLau2014,JZhang2014,CFan2015}. These approaches were originally proposed to reduce the implementational and computational complexity. In contrast, in this paper, we use channel sparsification as a tool to identify the most interfering users in the received signal of each RRH. We only assign orthogonal training sequences to the most interfering users and ignore the rest, hence the name \emph{local orthogonality}.
	
	We basically follow the channel sparsification approach in \cite{CFan2015}. The only difference is that here the $l_\infty$ norm\footnote{For a vector $\mathbf{x}=[x_1,x_2,\dots, x_N]$, $\|\mathbf{x}\|_\infty = \max\{|x_1|,\dots,|x_N|\}$. The reason that we adopt $l_\infty$ norm is to ease our analysis on the minimum training length detailed in Section \ref{sec:Asymptotic}, where the random geometric theory developed in \cite{Martin1997} directly applies. Other types of norm can also be used but the required random geometric theory is different. We refer interested readers to \cite{RGG} for details.} is adopted as a measure of the distance between two nodes. Specifically, the channel sparsification is to ignore relatively weak channel links based on the following criteria:
	\begin{equation} \label{equ:sparsification}
		\tilde{h}_{i,k} = \left\{
		\begin{array}{l}
			h_{i,k}, \ \|\mathbf{b}_i - \mathbf{u}_k \|_\infty < r\\
			0,   \quad  \  \text{otherwise}\\
		\end{array}
		\right.
	\end{equation}
	where $r$ is a predefined threshold, and $\mathbf{u}_k\in \mathbb{R}^{1\times 2}$ and $\mathbf{b}_i\in \mathbb{R}^{1\times 2}$ denote the coordinates of user $k$ and RRH $i$, respectively. 
	
	We now present graphical illustrations of CRANs after channel sparsification. Denote by 
	\begin{align}
		\mathcal{U} &= \{1,\cdots,K\}\\
		\intertext{the set of user indexes and by}
		\mathcal{B}&=\{1,\cdots,N\}
	\end{align}
	the set of RRH indexes. Define
	\begin{align}
		\mathcal{B}_k   &\triangleq \text{the set of RRHs serving user $k$, for $k\in \mathcal{U}$}\\
		\mathcal{U}_i   &\triangleq \text{the set of users served by RRH $i$, for $i\in \mathcal{B}$}\\
		\mathcal{U}_i^c &\triangleq \text{the complement of $\mathcal{U}_i$, for $k\in \mathcal{U}$}.
	\end{align}
	The user associations with RRHs after channel sparsification are illustrated in Fig. \ref{Fig:Channel_Sparsification}(a), where each user is connected to an RRH by an arrow if the distance between the user and the RRH is below the threshold $r$. Alternatively, the system after channel sparsification can also be represented as a bipartite graph shown in Fig. \ref{Fig:Channel_Sparsification}(b), where each black node represents an RRH, and each white node represents a user.
		
	\begin{figure}[!t]
		\centering
		\includegraphics[width = 0.31\textwidth]{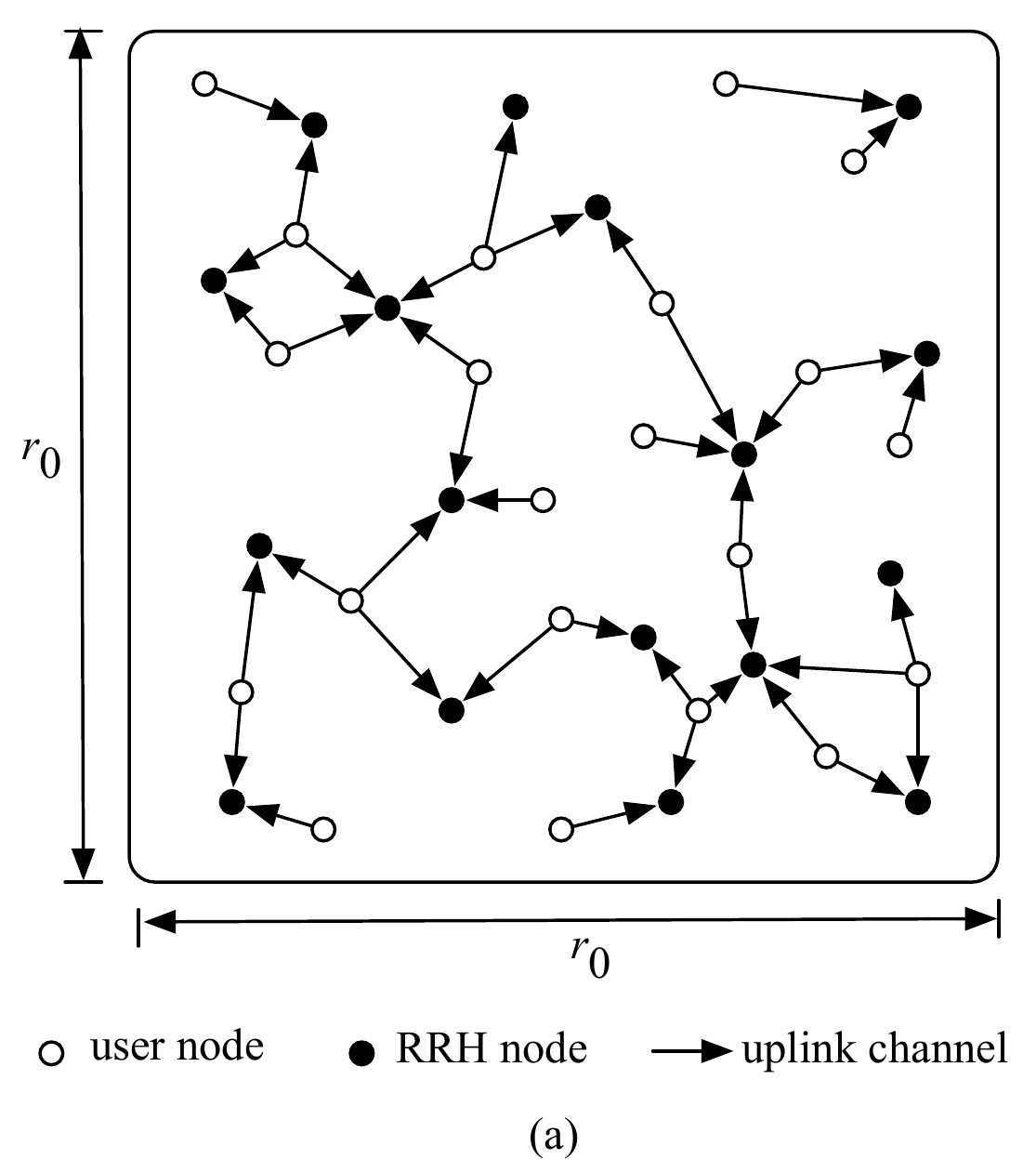}
	\hspace{50pt}
		\includegraphics[width = 0.28\textwidth]{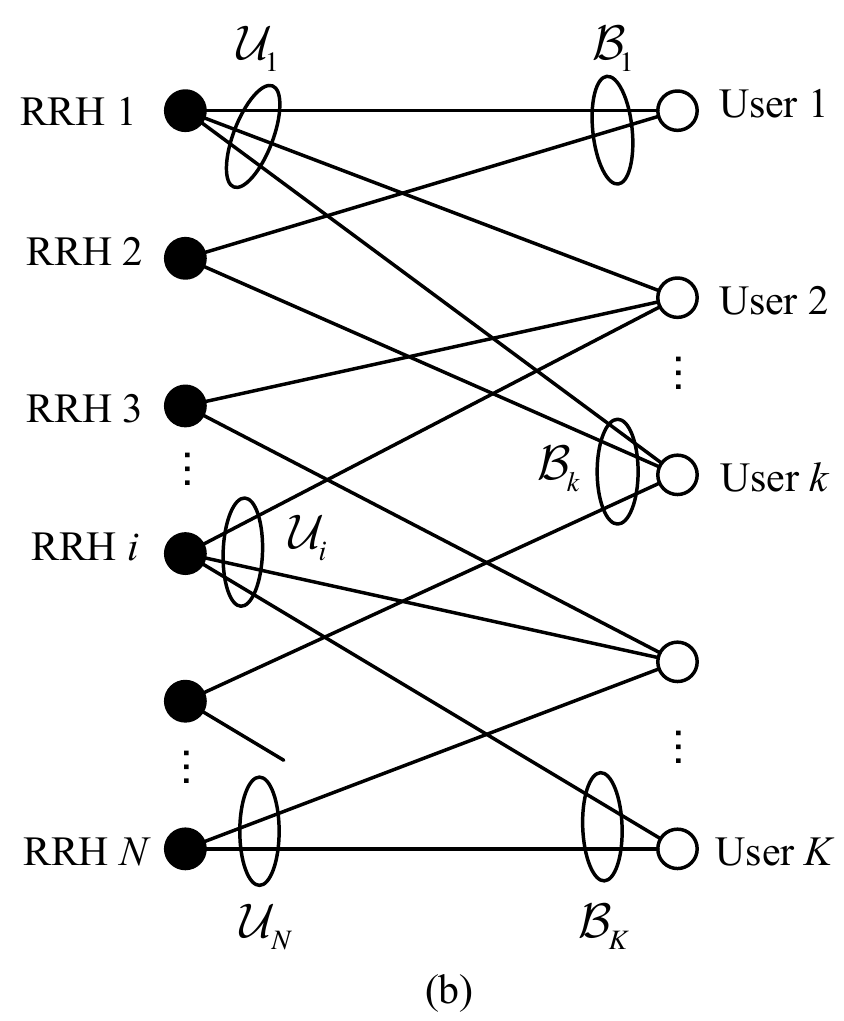}
		\caption{(a) User association after channel sparsification. (b) The bipartite graph representation of user association after channel sparsification.}\label{Fig:Channel_Sparsification}
	\end{figure}
	
	With the above channel sparsification, the received signal at RRH $i$ in \eqref{equ:Training Model} can be rewritten as
	\begin{equation}\label{equ:Training Model2}
		\mathbf{y}_{i}^\text{p} = \sum_{k\in \mathcal{U}_i} h_{i,k} \gamma_{i,k} \mathbf{x}_{k}^\text{p} + \sum_{k\in \mathcal{U}_i^c} h_{i,k} \gamma_{i,k} \mathbf{x}_{k}^\text{p} + \mathbf{z}_{i}^\text{p}.
	\end{equation}
	We aim to minimize the multiuser interference in the first term of the right-hand-side of \eqref{equ:Training Model2}, while the interference in the second term is ignored as it is much weaker than the one in the first term. To this end, the training sequences $\{\mathbf{x}_{k}^\text{p},k\in \mathcal{U}_i\}$ should be mutually orthogonal for any $i\in \mathcal{B}$. This gives a formal definition of \emph{local orthogonality}.  
	
\subsection{Problem Statement}\label{subsec:ProbStatement}
	The goal of this work is to design training sequences with the shortest length that preserve local orthogonality. This problem is formulated as
	\begin{subequations}\label{equ:pilot prob}
		\begin{align}
			\min_{\{\mathbf{x}_{k}^\text{p}\}} \quad &\alpha \\
			\text{s.t.} \quad         &\mathbf{x}_{k}^\text{p} \perp \mathbf{x}_{k^\prime}^\text{p}, \forall k\ne k^\prime, k \text{ and } k^\prime \in \mathcal{U}_i, \forall i\in \mathcal{B}\label{equ:local orthogonal}\\
			                          &\mathbf{x}_{1}^\text{p},\cdots,\mathbf{x}_{K}^\text{p} \in \mathbb{C}^{1\times \alpha T}
		\end{align}
	\end{subequations}
	where $\alpha$ defined in Section \ref{subsec:training phase} is the time-splitting factor for the training phase; $\mathbf{a} \perp \mathbf{b}$ means that $\mathbf{a}$ is orthogonal to $\mathbf{b}$. 
	
	It is not easy to tackle the problem in \eqref{equ:pilot prob} directly, partly due to the fact that the dimension of the search space varies with $\alpha$. In the following, we will solve \eqref{equ:pilot prob} by converting it to a graph coloring problem. In addition, the optimal $\alpha$ is a random variable depending on the random locations of RRHs and users. We will characterize the asymptotic behavior of the optimal $\alpha$ as the network size goes to infinity.
	
	The problem formulation in \eqref{equ:pilot prob} is for uplink channel estimation. We emphasize that the training design for the uplink directly carries over the downlink by swapping the roles of users and RRHs. That is, in the uplink phase, the training sequences are transmitted by users and the local orthogonality is preserved at the RRH side, while in the downlink, the training sequences are transmitted by RRHs and the local orthogonality is preserved at the user side. Moreover, if channel reciprocity is assumed, channel training is only required once, either at the RRH side or at the user side. Therefore, we henceforth focus on the training design for the uplink. 
	
\section{Training Sequence Design}\label{sec:pilot design}
	In this section, we solve problem \eqref{equ:pilot prob} based on graph coloring. We first formulate a graph coloring problem that is equivalent to problem \eqref{equ:pilot prob}.
\begin{figure}[!t]
	\centering
	\includegraphics[width = 0.29\textwidth]{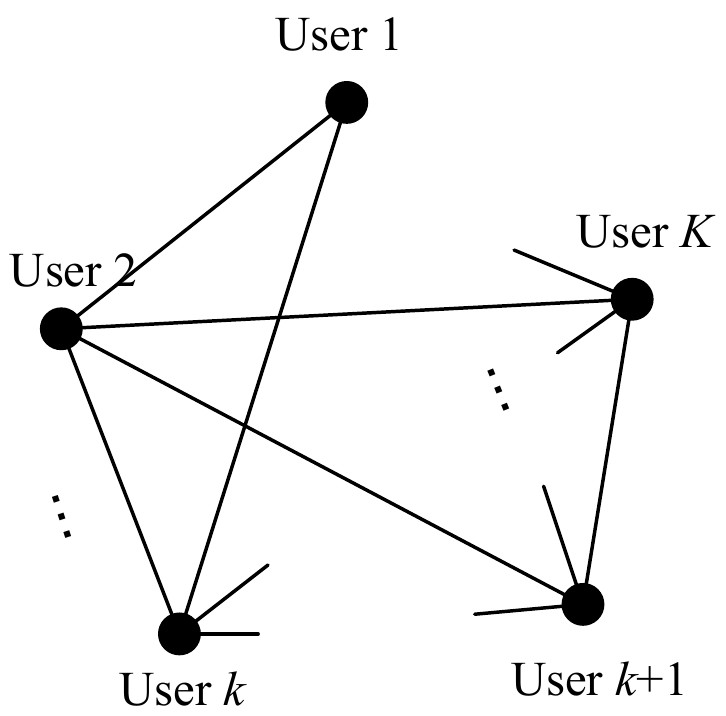}
	\caption{A new graphical representation of the graph in Fig. \ref{Fig:Channel_Sparsification}.}\label{Fig:Vertex_Coloring}
\end{figure}
 
	In Fig. \ref{Fig:Vertex_Coloring}, we define a new graph $G = \{\mathcal{U}, E\}$ with vertex set $\mathcal{U}$ and edge set $E$, where two users $k$ and $m$ in $\mathcal{U}$ are connected by an edge $e_{k,m} \in E$ if and only if they are served by a common RRH. Then, the edge set $E$ can be represented as $E = \{e_{k,m} | \mathcal{B}_k \cap \mathcal{B}_m \neq \emptyset, \forall k\neq m, k, m \in \mathcal{U} \}$. Denote by $c: \mathcal{U} \rightarrow \mathcal{C}$ a map from each user $k\in \mathcal{U}$ to a color $c(k)\in \mathcal{C}$. We then formulate the following vertex coloring problem over $G$:
	\begin{subequations}\label{equ:graph coloring}
		\begin{align}
			\min_{c}    \quad & |\mathcal{C}|\\
			\text{s.t.} \quad & c(k) \neq c(m), \text{ if } \mathcal{B}_k \cap \mathcal{B}_m \neq \emptyset, \forall k\neq m, k, m \in \mathcal{U}.\label{equ:graphEdgeConst}
		\end{align}
	\end{subequations}
	Note that the solution to \eqref{equ:graph coloring}, denoted by $\chi(G)$, is referred to as the chromatic number of the graph in Fig. \ref{Fig:Vertex_Coloring}. We further have the following result.
	\begin{theorem}\label{theorem:probEqu}
		The vertex coloring problem over $G$ in \eqref{equ:graph coloring} is equivalent to the training design problem in \eqref{equ:pilot prob}.
	\end{theorem}
	
	\begin{proof}
		Each color can be seen as an orthogonal training sequence. Then, the color set $\mathcal{C}$ can be mapped into a set of orthogonal training sequences $\{\mathbf{x}_{k}^\text{p}\}$. The cardinality of $\mathcal{C}$ equals to the number of orthogonal training sequences, i.e., $|\mathcal{C}| = \alpha T$. From \eqref{equ:graphEdgeConst}, the statement that any two vertices connected by an edge are colored differently is equivalent to the statement that any two users served by a common RRH transmit orthogonal training sequences. Then, as users in $\mathcal{U}_i$ are all served by RRH $i$, any two users in $\mathcal{U}_i$ must be connected by an edge in the new graph $G$. This is equivalent to say that the training sequences assigned to users in $\mathcal{U}_i$ are orthogonal to each other. Therefore, \eqref{equ:graphEdgeConst} is equivalent to \eqref{equ:local orthogonal}, which concludes the proof. 
	\end{proof}
	
	We now discuss solving the vertex coloring problem in \eqref{equ:graph coloring}. This is a well-known NP-complete problem \cite{Garey1979, Clark90}. Exact solutions can be found, e.g., using the algorithms proposed in \cite{Diaz2006, Sager1991}. However the running time of these algorithms is acceptable only when the corresponding graph has a relatively small size. For a large-size graph as in a CRAN, algorithms in \cite{Malaguti2008, Brelaz1979} are preferable to yield suboptimal solutions with much lower complexity. In this paper, we adopt in the simulations the \emph{dsatur algorithm}, which is a greedy-based low-complexity algorithm with near optimal performance \cite{Brelaz1979}. The dsatur algorithm dynamically chooses the vertex with the maximum saturation degree\footnote{For a vertex $u$, the (ordinary) degree is defined as the number of vertices connected to $u$; the saturation degree is defined as the number of vertices in distinct colors connected to $u$.} to color at each step, and prefers to use the existing colors to color the next vertex. For completeness, we describe the \emph{dsatur algorithm} in Table \ref{tab:algorithm1}.
	
	\begin{table}[!h]
    \caption{
    \textbf{Algorithm I}: The dsatur algorithm for graph coloring.}
    \label{tab:algorithm1}
    \centering
    \begin{tabular}{l}
    \hline
    \textbf{Initial:} $\mathcal{A} = \{k = 1,\dots,K\}$ and $\mathcal{C} = \emptyset$.\\
    \textbf{while} $\mathcal{A} \neq \emptyset$\\
    $\ \ $ Step 1: Select $k\in \mathcal{A}$ with maximum saturation degree.\\
    $\quad$$\quad$ If two vertexes have the same maximum saturation\\
    $\quad$$\quad$ degree, choose the one with maximum ordinary\\
    $\quad$$\quad$ degree.\\
    $\ \ $ Step 2: Color $k$ greedily so that\\
    $\quad\quad c(k) = \min \{i\in \mathcal{C}| c(m) \neq i, \forall m \in \{m| e_{k,m}\in E\} \}$\\
    $\quad$$\quad$ \textbf{if} $c(k) = \emptyset$ \\
    $\quad$$\quad$$\quad$ $c(k) = |\mathcal{C}|+1$ \\
    $\quad$$\quad$$\quad$ $\mathcal{C} = \mathcal{C}\cup \{c(k)\}$\\
    $\quad$$\quad$ \textbf{end if} \\
    $\ \ $ Step 3: $\mathcal{A} = \mathcal{A} \backslash k$\\
    \textbf{end while}\\
    \hline
    \end{tabular}
    \end{table}
	
	We now discuss the construction of training sequences $\{\mathbf{x}_{k}^\text{p}\}$ for problem \eqref{equ:pilot prob} based on the coloring pattern $c$ and the chromatic number $\chi(G)$ obtained from solving \eqref{equ:graph coloring}. We first generate $\chi(G)$ orthonormal training sequences of length $\chi(G)$, i.e. $\tilde{\mathbf{x}}_{i}^\text{p} (\tilde{\mathbf{x}}_{i}^\text{p})^\text{H} = 1,\forall i, \text{ and } \tilde{\mathbf{x}}_{i}^\text{p} (\tilde{\mathbf{x}}_{j}^\text{p})^\text{H} = 0,\forall i,j \in \{1,\dots,\chi(G)\} \text{ with } i\neq j$. Then, the training sequence $\mathbf{x}_{k}^\text{p}$ for user $k$ is scaled to meet the power constraint in \eqref{equ:Training PC}, i.e. $\mathbf{x}_{k}^\text{p} = \sqrt{\chi(G)\beta_k P_0} \tilde{\mathbf{x}}_{c(k)}^\text{p},\ k \in \mathcal{U}$.
	
\section{Optimal Training Length}\label{sec:Asymptotic}
	In the preceding section, we proposed a training design algorithm for the problem in \eqref{equ:pilot prob}, and obtained that the minimum training length $\alpha T$ is given by the chromatic number $\chi(G)$. In this section, we focus on the behavior of the training length as the network size increases. We show that the training length scales at most at the rate of $O(\ln K)$ as the network size goes to infinity under the assumption of a fixed user density.   
	
\subsection{Graph with Infinite RRHs}
	Our analysis is based on the theory of random geometric graph. Recall from \eqref{equ:sparsification} that the edge generation of graph $G$ in Fig. \ref{Fig:Vertex_Coloring} follows the rule below:
	\begin{equation}
		E = \{e_{k,m} | \exists i, \|\mathbf{u}_k-\mathbf{b}_i\|_\infty <r \text{ and } \|\mathbf{u}_m-\mathbf{b}_i\|_\infty <r\}.
	\end{equation}
	This is unfortunately different from the edge-generation rule of a random geometric graph \cite{Martin1997}. To circumvent this obstacle, we introduce a new graph as follows:
	\begin{equation}\label{equ:gvinfty}
		G^\infty =\{\mathcal{U},E^\infty\} \text{ with } E^\infty = \{e_{k,m}^\infty | \|\mathbf{u}_k-\mathbf{u}_m\|_\infty <2r \}.
	\end{equation}
	We have the following result.
	\begin{lemma}\label{lemma1}
		The graph $G$ is a subgraph of $G^\infty$.
	\end{lemma}
	\begin{proof}
		Since $G$ and $G^\infty$ have a common vertex set $\mathcal{U}$, we only need to prove $E \subseteq E^\infty$. For any $e_{k,m} \in E$, we have $\|\mathbf{u}_k-\mathbf{u}_m\|_\infty =\|\mathbf{u}_k-\mathbf{b}_i + \mathbf{b}_i-\mathbf{u}_m\|_\infty \leq \|\mathbf{u}_k-\mathbf{b}_i\|_\infty + \|\mathbf{u}_m-\mathbf{b}_i\|_\infty < 2r$. Hence, $e_{k,m} \in E^\infty$. 
	\end{proof}
	
	Graph $G^\infty$ can be seen as a supergraph\footnote{A supergraph is a graph formed by adding vertices, edges, or both to a given graph. If $H$ is a subgraph of $G$, then $G$ is a supergraph of $H$.} of $G$ with infinite RRHs, i.e., there always exists an RRH located between two users provided that the distance between these two users does not exceed $2r$. Since $G$ is a subgraph of $G^\infty$, the chromatic number $\chi(G^\infty)$ serves as an upper bound for $\chi(G)$.
		
\subsection{Asymptotic Behavior of the Training Length}
	We are now ready to characterize the asymptotic behavior of the training length as $K\rightarrow \infty$. Denote by $\delta = \frac{K}{r_0^2}$ the user density in the service area. We have the following theorem.

	\begin{theorem}\label{theorem:Asym}
		As $K\rightarrow\infty$ and $\frac{\delta r^2}{\ln K}\rightarrow \rho$, the minimum training length $\alpha T$ preserving local orthogonality satisfies 
		\begin{equation}\label{equ:AsymChrom}
			\limsup_{K\rightarrow \infty} \left( \frac{\alpha T}{\delta r^2} \right) \leq 4f^{-1}\left(\frac{1}{4\rho}\right),\ a.s.,
		\end{equation}
		where $\rho \in (0, +\infty)$ is a predetermined constant, and $f^{-1}$ is the inverse function of $f(x)$ over the domain $[1,+\infty)$ with $f(x)$ defined as below
		\begin{equation}\label{equ:f_fun}
			f(x) = 1-x+x\ln x, \quad x>1
		\end{equation}
		and $a.s.$ stands for almost surely.
	\end{theorem}
	
	\begin{proof} 
	From Lemma \ref{lemma1}, $G$ is a subgraph of $G^\infty$, and $\chi(G)$ is upper bounded by $\chi(G^\infty)$. As the minimum $\alpha T$ is given by $\chi(G)$, it suffices to show
	\begin{equation}
		\limsup_{K\rightarrow \infty} \left( \frac{\chi(G^\infty)}{\delta r^2} \right) \leq 4f^{-1}\left(\frac{1}{4\rho}\right),\ a.s.
	\end{equation}

	Denote by $\omega(G^\infty)$ the clique number\footnote{A complete graph is a graph where every pair of distinct vertices are connected by a unique edge. The clique number is defined as the number of vertices of the largest complete subgraph of $G^\infty$.} of $G^\infty$. From Theorem 2 in \cite{Martin1997}, we have\footnote{To invoke Theorem 2 in \cite{Martin1997}, the distance threshold in \cite{Martin1997} is set to $\frac{2r}{r_0}$.}
		\begin{equation}\label{equ:cliqueLimit}
			\limsup_{K\rightarrow \infty} \left( \frac{\omega(G_v^\infty)}{\delta r^2} \right) \leq 4f^{-1}\left(\frac{1}{4\rho}\right),\ a.s.
		\end{equation}
		The clique number $\omega(G^\infty)$ serves as a lower bound for $\chi(G^\infty)$ in general. However, as $K\rightarrow \infty$, the clique number almost surely converges to the chromatic number \cite{McDiarmid2005}, i.e.
		\begin{equation}\label{equ:asyequ}
			\lim_{K\rightarrow \infty} \frac{\chi(G^\infty)}{\omega(G^\infty)} = 1,\ a.s.
		\end{equation}
		Combining \eqref{equ:cliqueLimit} and \eqref{equ:asyequ}, we obtain 
	\begin{equation}
		\limsup_{K\rightarrow \infty} \left( \frac{\chi(G^\infty)}{\delta r^2} \right) \leq 4f^{-1}\left(\frac{1}{4\rho}\right),\ a.s.
	\end{equation}
	This completes the proof.
	\end{proof}

	\emph{Remark 1}: We note that the upper bound of the minimum training length in Theorem \ref{theorem:Asym} is independent of the number of RRHs $N$. For a finite $N$, the minimum training length increases with $N$, and is upper bounded by $\chi(G^\infty)$. 

	\emph{Remark 2}: From Theorem \ref{theorem:Asym}, with given $\rho$ and $r$, the minimum training length to preserve local orthogonality scales in the order of $\ln K$, provided that the user density $\delta = O(\ln K)$.
	
	From Theorem \ref{theorem:Asym}, we also claim the following result.
	
	\begin{corollary}
		As $K\rightarrow \infty$ and $r_0^2 \rightarrow \infty$ with $\delta = \frac{K}{r_0^2}$ fixed, the minimum training length to preserve local orthogonality scales at most in the order of $\ln K$.
	\end{corollary}
	
	\begin{proof}
		Note that a graph with a fixed user density $\delta$ can be generated by randomly deleting users from a graph with user density $\delta = O(\ln K)$. This implies that the minimum training length for the case of $\delta$ fixed is upper bounded by that of the case of $\delta = O(\ln K)$. From Theorem \ref{theorem:Asym}, when $\delta = O(\ln K)$, the minimum training length scales in the order of $\ln K$. Therefore, the minimum training length for a fixed $\delta$ scales at most in the order of $\ln K$, which completes the proof.
	\end{proof}
		
	\emph{Remark 3}: Corollary 1 indicates that the minimum training length of our proposed training design scheme is moderate even for a large-scale CRAN.

\subsection{Further Discussions}
	In the above discussions, we assume that the chromatic number of $G$ can be determined accurately. However, the algorithms proposed in \cite{Diaz2006, Sager1991} to find the chromatic number always have high computation complexity and are not suitable for practical systems. As mention in Section \ref{sec:pilot design}, a greedy algorithm called \emph{the dsatur algorithm} is applied in this paper, which cannot be guaranteed to achieve the chromatic number but with relatively low complexity. However, the problem arises whether the number of colors used by a suboptimal coloring algorithm still scales in the order of $\ln K$. For this problem, we have the following theorem.
	\begin{theorem}\label{proposition1}
		The training length $\alpha T$ determined by the dsatur algorithm in Table \ref{tab:algorithm1} scales at most in the order of $\ln K$ as $K\rightarrow\infty$ and $\frac{\delta r^2}{\ln K}\rightarrow \rho$.
	\end{theorem}
	
	\begin{proof}
		Revisit the dsatur algorithm in Table \ref{tab:algorithm1}. In each coloring step, the vertex to be colored prefers to use an existing color. Then, the number of colors in every step is upper bounded by $\Delta(G)+1$, where $\Delta(G)$ is the maximum degree of $G$. Therefore, the number of colors used by the dsatur algorithm is upper bounded by $\Delta(G)+1$. We also see from Lemma \ref{lemma1} that $\Delta(G) \leq \Delta(G^\infty)$. Therefore, it suffices to characterize the behavior of $\Delta(G^\infty)$. From Theorem 1 in \cite{Martin1997}, we have
		\begin{equation}
			\limsup_{K\rightarrow \infty} \left( \frac{\Delta(G^\infty)}{\delta r^2} \right) \leq 16f^{-1}\left(\frac{1}{16\rho}\right),\ a.s.
		\end{equation}
		where $f^{-1}$ is given in \eqref{equ:f_fun}. Together with $\delta r^2 = O(\ln K)$, we see that $\Delta(G^\infty)$ scales in the order of $\ln K$. Thus, $\Delta(G)$ and the number of colors given by the dsatur algorithm scales at most in the order of $\ln K$, which completes the proof.
	\end{proof}
	
	Fig. \ref{Fig:Chromatic} gives the numerical results to verify our analysis. Both the chromatic number of $G^\infty$ and $G$ with $N = 1000$ are included for comparision. In simulation, the chromatic number is found by the dsatur algorithm and averaged over $1000$ random realizations. From Fig. \ref{Fig:Chromatic}, the training length $\alpha T$ is strictly upper bounded by the theoretical result given in \eqref{equ:AsymChrom}. We also note that the output of the dsatur algorithm for $G^\infty$ may be slightly greater than the upper bound given in \eqref{equ:AsymChrom}. This is due to the use of the suboptimal coloring algorithm in simulation. From Fig. \ref{Fig:Chromatic}, we see that the simulated minimum training length is close to the upper bound given in \eqref{equ:AsymChrom}.
	
	\begin{figure}[!h]
		\centering
		\includegraphics[width = 0.48\textwidth]{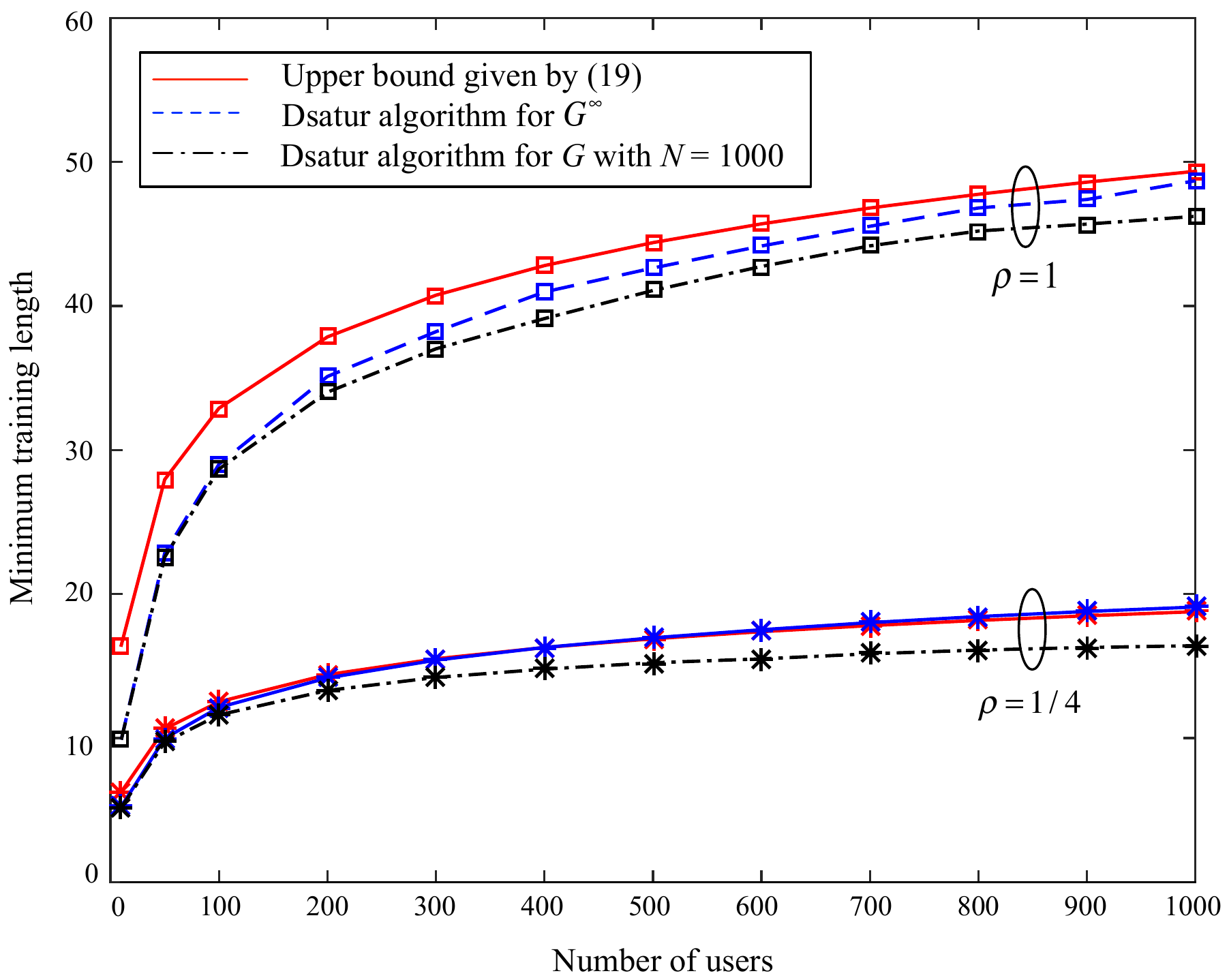}
	 	\caption{The minimum training length against the number of users $K$ with $N = 1000$.}\label{Fig:Chromatic}
	\end{figure}

\section{Practical Design}\label{sec:PracticalDesign}
	In this section, we evaluate the proposed training-based scheme using the information throughput as the performance measure. The throughput expression can be found in Appendix \ref{appendixI}.
	
\subsection{Refined Channel Sparsification}
	We first show that the channel sparsification criteria in \eqref{equ:sparsification} can be refined to improve the system throughput while keeping the minimum training length unchanged and still preserving local orthogonality. Recall the received signal after channel sparsification shown in \eqref{equ:Training Model2}. Due to local orthogonality, for each RRH $i$, the training sequences $\{\mathbf{x}_{k}^\text{p},k\in \mathcal{U}_i \}$ in the first term of \eqref{equ:Training Model2} are orthogonal and the channel coefficients $\{h_{i,k},k\in \mathcal{U}_i \}$ can be estimated, while the channel coefficients $\{h_{i,k},k\in \mathcal{U}_i^c \}$ in the second term are not estimated. However, due to the randomness of user locations, a different RRH usually serves a different number of users, which can be seen from Fig. \ref{Fig:UserNumber_pdf}. For a certain RRH, some training sequences in the second term may be orthogonal to the training sequences $\{\mathbf{x}_{k}^\text{p},k\in \mathcal{U}_i \}$ in the first term. The corresponding channel coefficients in the second term can also be estimated at RRH $i$, so as to mitigate the interference. Recall that adding more users to $\{\mathcal{U}_i \}$ is equivalent to adding more edges to the bipartite graph in Fig. \ref{Fig:Channel_Sparsification}. Thus, to suppress interference, it is desirable to add more edges to the bipartite graph without compromising local orthogonality. 
	\begin{figure}[!t]
		\centering
		\includegraphics[width = 0.48\textwidth]{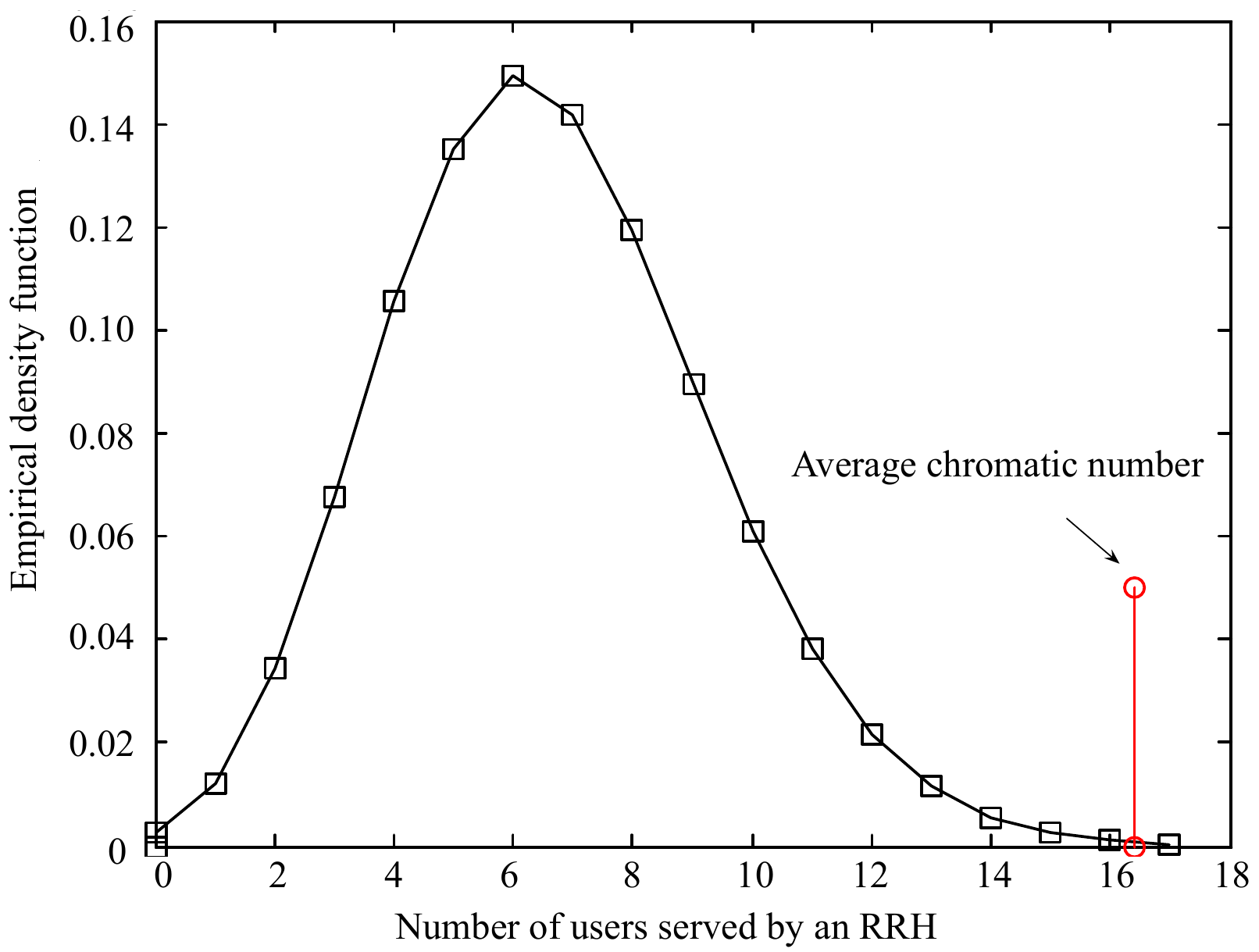}
	 	\caption{Empirical density function of the number of users served an RRH with $N = 1000$ and $K = 1000$. Every simulated point is averaged over 1000 channel samples. The average chromatic number is also included for reference.}\label{Fig:UserNumber_pdf}
	\end{figure}
	\begin{figure}[!t]
		\centering
		\includegraphics[width = 0.31\textwidth]{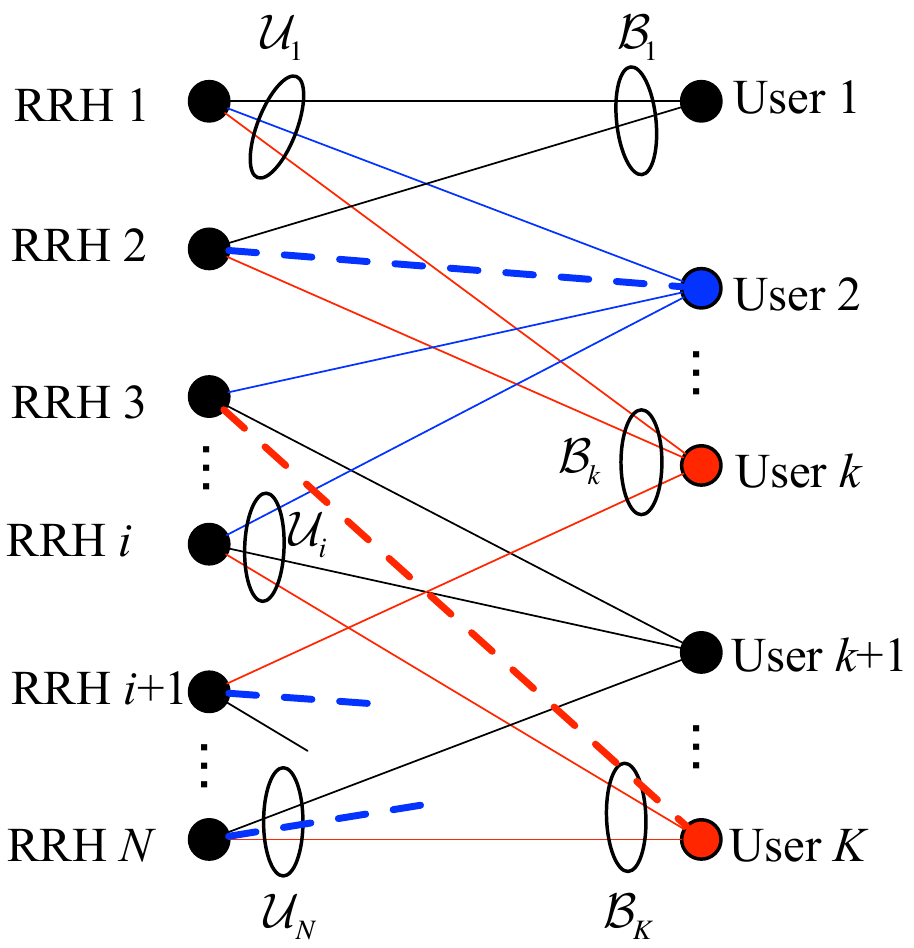}
	 	\caption{An example to illustrate the refined channel sparsification.}\label{Fig:RefindSpa}
	\end{figure}
	
	We now show how to add more users to each RRH. Denote by $c$ and $\chi(G)$ the coloring pattern and the chromatic number obtained by solving \eqref{equ:graph coloring}, respectively. Based on $c$, the user set $\mathcal{U}$ can be partitioned into $\chi(G)$ subsets each with a different color. Then, each RRH $i$ chooses the closest users colored differently from the users in $\mathcal{U}_i$. An example to this procedure is illustrated in Fig. \ref{Fig:RefindSpa}, where solid lines represent the user association after channel sparsification following the criteria in \eqref{equ:sparsification}, and the dash lines denote the new edges added to RRHs. The coloring pattern $c$ is obtained by solving problem \eqref{equ:graph coloring}. 
			
	Clearly, the refined channel sparsification in Fig. \ref{Fig:RefindSpa} still preserves local orthogonality. Hence, the analytical results in Section \ref{sec:Asymptotic} are still applicable to the refined channel sparsification and the minimum training length still follows $O(\ln K)$. Furthermore, the number of users associated with each RRH is exactly $\chi(G)$, so as to minimize the interference.  
	
\subsection{Numerical Results}\label{subsec:Numerical}
	Numerical results are presented to demonstrate the performance of our proposed scheme. Users and RRHs are uniformly scattered in a square area with $r_0=100$m, i.e., RRHs and users are uniformly distributed over a $100\text{m} \times 100\text{m}$ region. The pathloss exponent is $\eta = 3.5$ and the power ratios $\beta_1 = \beta_2 = \cdots = \beta_K = 1$. The channel coherence time is fixed at $T = 100$.
	
	Performance comparison among different schemes is given in Fig. \ref{Fig:PerfComp} with $\rho = 0.5$ ($\rho$ is defined in Theorem \ref{theorem:Asym}) and $N = K = 1000, 600, 300$ marked on the curves. From Fig. \ref{Fig:PerfComp}, we see that the refined channel sparsification method outperforms the channel sparsification in \eqref{equ:sparsification} by about $10\%$. We also compare the proposed scheme with the random pilot scheme, where the pilot lengths of the two schemes set to be equal for a fair comparison. We see that our proposed pilot design scheme has over $20\%$ performance enhancement over the random pilot scheme. The scheme in \cite{Coldrey2007} is also included for comparison. Specifically, it was shown in \cite{Coldrey2007} that, for a conventional training-based MIMO, the optimal $K$ and $\alpha$ respectively converge to $\alpha T$ and $1/2$ at high SNR. Thus, we randomly choose $T/2$ active users from all the $K$ users and assign an orthogonal pilot sequence to each user. We refer to this scheme as ``globally orthogonal scheme''. We see from Fig. \ref{Fig:PerfComp} that our proposed scheme significantly outperforms the orthogonal scheme. The reason is that the orthogonal scheme spend more time resource in channel training.  
	\begin{figure}[!t]
		\centering
		\includegraphics[width = 0.48\textwidth]{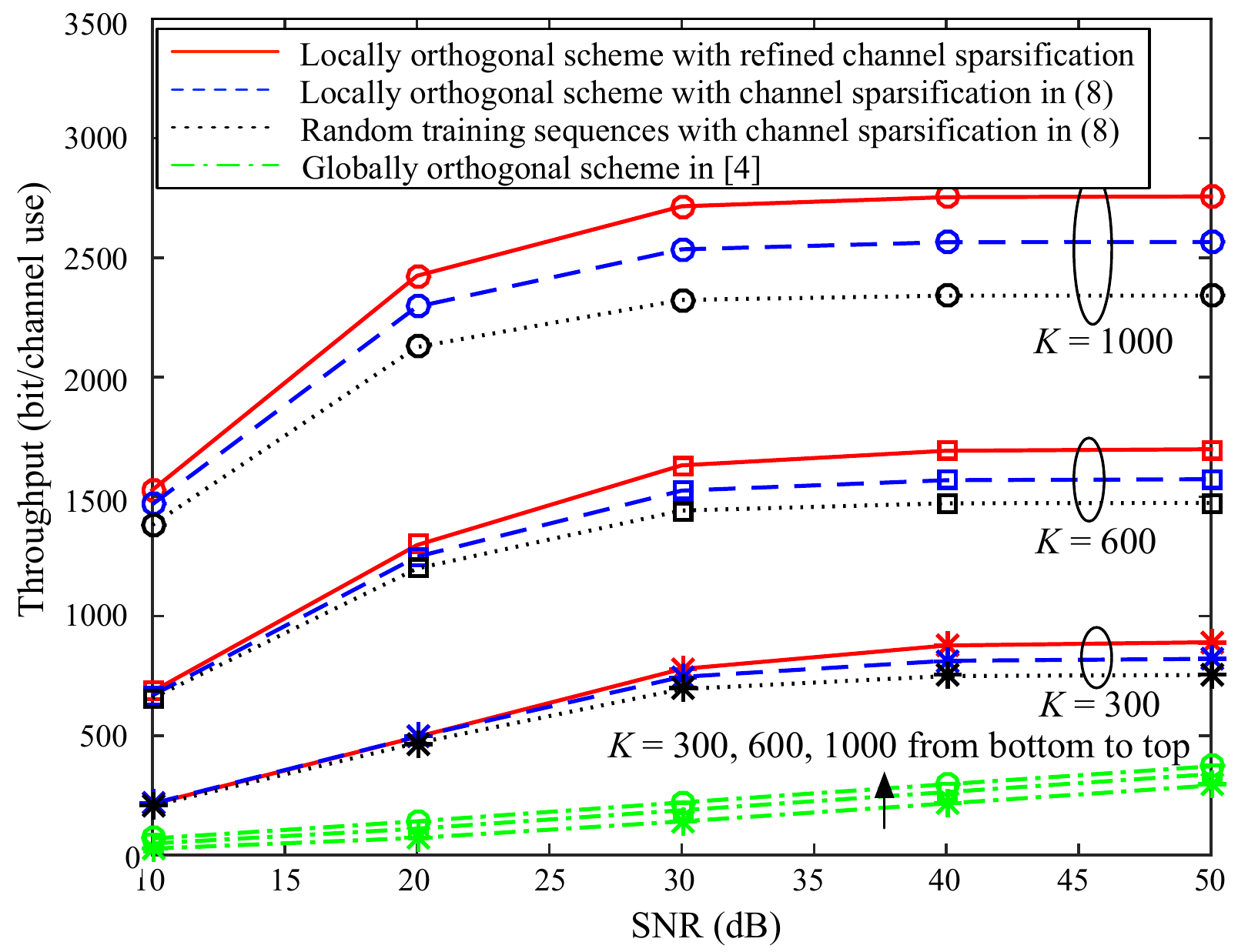}
	 	\caption{Performance comparison between the new and the old channel sparsification criteria with $\rho = 0.5$ and $N = K = 600,300$. SNR = $P_0/N_0$ and $r_0=100$m.}\label{Fig:PerfComp}
	\end{figure}
	\begin{figure}[!t]
		\centering
		\includegraphics[width = 0.49\textwidth]{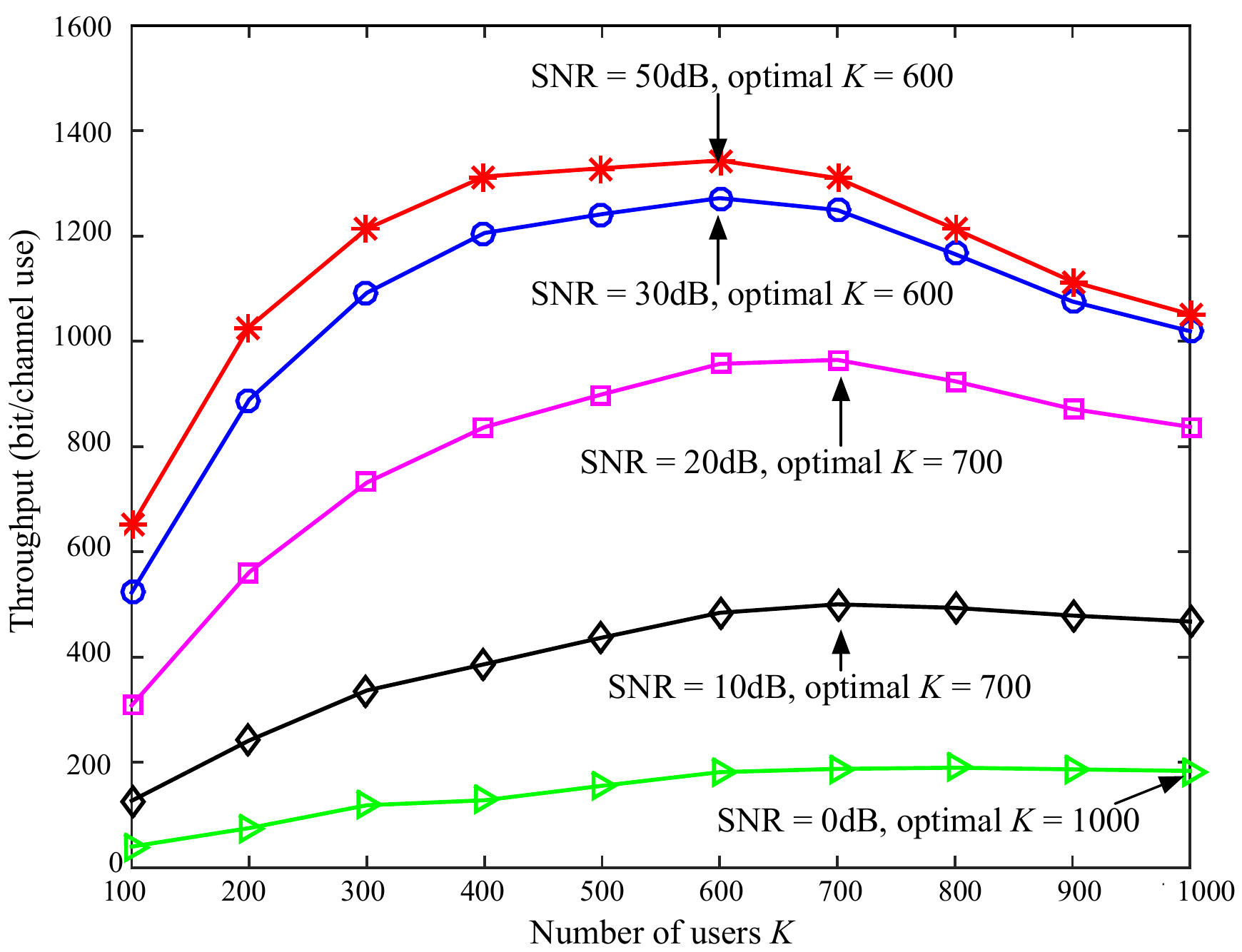}
	 	\caption{The system throughput versus the number of users. The number of RRHs $N=500$. The distance threshold for channel sparsification is set to $r = 10$m and the side length $r_0=100$m.}\label{Fig:PerfvsK}
	\end{figure}
	\begin{figure}[!t]
		\centering
		\includegraphics[width = 0.49\textwidth]{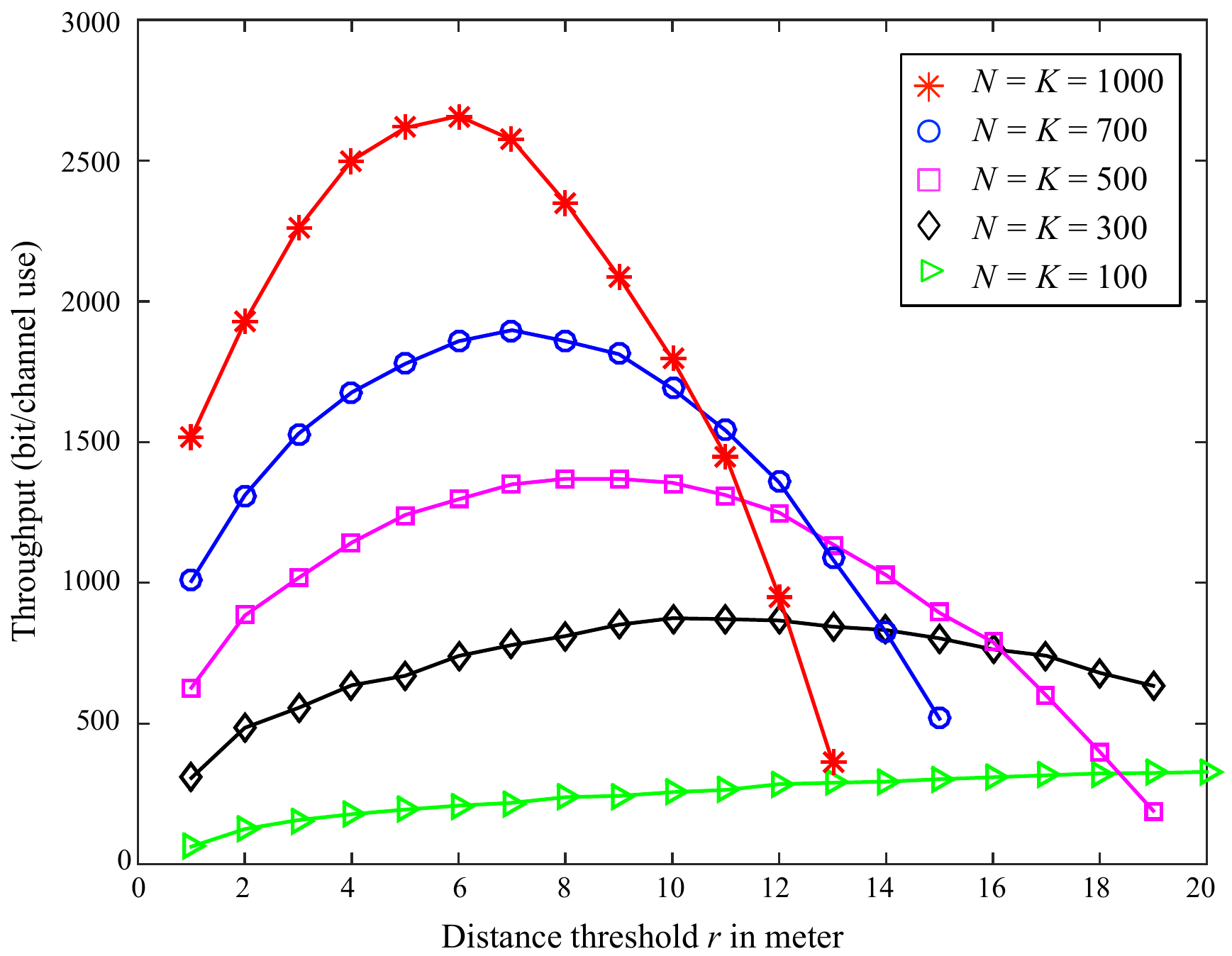}
	 	\caption{The system throughput versus the distance threshold $r$. The signal to noise ratio SNR = $P_0/N_0 = 50$dB, and the side length $r_0=100$m.}\label{Fig:PerfvsRth}
	\end{figure}
	
	The system throughput against the number of users is given in Fig. \ref{Fig:PerfvsK}. We see that, for given $N$, $r$, and SNR = $P_0/N_0$, there is a tradeoff in maximizing the system throughput over $K$. On one hand, the total received power at each RRH increases with $K$, and so is the information rate for the data transmission phase. On the other hand, a larger $K$ implies more channel coefficients to be estimated, and hence increases the required training length. Therefore, a balance needs to be stroke in optimizing $K$. For example, for SNR = 30dB, the optimal $K$ occurs at $K = 600$.
	
	A larger $K$ implies higher received signal power at each RRH and at the same time, higher interference power at each RRH (as by channel sparsification the signals from far-off users are treated as interference). In the low SNR region, the interference is overwhelmed by the noise. Therefore, the more users, the higher system throughput, as seen from the curve for SNR = 0dB in Fig. \ref{Fig:PerfvsK}. In the high SNR region, the interference dominates the noise. Then, a larger $K$ does not necessarily translates to a higher throughput. 
	
	Fig. \ref{Fig:PerfvsRth} illustrates the system throughput against the distance threshold $r$ with SNR = 50dB. We see that there is an optimal distance threshold to maximize the throughput for fixed $N$ and $K$. The reason is explained as follows. With a high threshold $r$, more interference from neighboring users are eliminated in channel estimation, and thus better throughput is achieved. However, at the same time, increasing $r$ implies increasing the training length, as more training sequences are required to be orthogonal at each RRH. Therefore, there is an optimal $r$ for throughput maximization, as illustrated in Fig. \ref{Fig:PerfvsRth}.  
	
\section{Conclusion}
	In this paper, we considered training-based channel estimation for CRANs with $N$ RRHs and $K$ users. We introduced the notion of local orthogonality and formulated the training design problem so as to find the minimum length of training sequences that preserve local orthogonality. A training design scheme based on graph coloring was also proposed. We further showed that the training length is $O(\ln K)$ almost surely as $K\rightarrow \infty$. Therefore, the proposed training design can be applied to a large-size CRAN satisfying local orthogonality at the cost of an acceptable training length.
	
	In this paper, we mainly focus on minimizing the training length. The joint optimization of the training sequences and the data transmission scheme to maximize the system throughput will be an interesting topic for future research.
	
\appendices
\section{Mutual Information Througput}\label{appendixI}
	In this appendix, we derive a throughput lower bound for the training-based CRAN scheme described in Section \ref{sec:preliminaries}. With channel sparsification, the received signal at RRH $i$ in \eqref{equ:Training Model} can be rewritten as
	\begin{equation}\label{equ:Training Model3}
		\mathbf{y}_{i}^\text{p} = \sum_{k\in \mathcal{U}_i} h_{i,k} \gamma_{i,k} \mathbf{x}_{k}^\text{p} +\sum_{k\in \mathcal{U}_i^c} h_{i,k} \gamma_{i,k} \mathbf{x}_{k}^\text{p} + \mathbf{z}_{i}^\text{p},i\in \mathcal{B}.
	\end{equation}
	Each RRH estimates the channel $\{h_{i,k},k\in \mathcal{U}_i\}$ based on $\mathbf{y}_i^\text{p}$ and $\mathbf{x}_{k}^\text{p}$. The minimum mean-square error (MMSE) estimator \cite{Kay2001} of RRH $i$ for user $k$ is given by
	\begin{equation}
		\mathbf{w}_{i,k} = \gamma_{i,k} \mathbf{x}_{k}^\text{p} \bigg(\sum_{k^\prime \in \mathcal{U}_i} \gamma_{i,k^\prime}^2 (\mathbf{x}_{k^\prime}^\text{p})^\text{H} \mathbf{x}_{k^\prime}^\text{p} + N_0 \mathbf{I} \bigg)^{-1},\  k \in \mathcal{U}_i.
	\end{equation}
	where $\mathbf{w}_{i,k} \in \mathbb{C}^{1 \times \alpha T}$. Then the estimation of $h_{i,k}$ denoted by $\hat{h}_{i,k}$ is given by
	\begin{equation}
		\hat{h}_{i,k}=\mathbf{y}_{i}^\text{p} \mathbf{w}_{i,k}^{\text{H}}, \ k\in \mathcal{U}_i, \ i\in \mathcal{B}
	\end{equation}
	where both $\mathbf{y}_{i}^\text{p}$ and $\mathbf{w}_{i,k}$ are row vectors. For $k\in \mathcal{U}_i^c, i\in \mathcal{B}$, the channel estimate of $h_{i,k}$ is set to $0$, i.e.,
	\begin{equation}\label{equ:estHc}
		\hat{h}_{i,k} = 0,\ k\in \mathcal{U}_i^c, \ i\in \mathcal{B}.
	\end{equation}
	Denote by $\text{MSE}_{i,k}$ the corresponding mean square error (MSE) of RRH $i$ for user $k$. Then
	\begin{subequations}
		\begin{align}
			\text{MSE}_{i,k} &= \mathbb{E} \Big[|h_{i,k} - \hat{h}_{i,k}|^2 \Big]\nonumber \\
			          &= 1-\gamma_{i,k}\mathbf{x}_{k}^\text{p}\mathbf{w}_{i,k}^{\text{H}} + \mathbf{w}_{i,k}\Big(\sum_{k^\prime\in \mathcal{U}_i^c} \gamma_{i,k^\prime}^2 (\mathbf{x}_{k^\prime}^\text{p})^\text{H} \mathbf{x}_{k^\prime}^\text{p} \Big) \mathbf{w}_{i,k}^{\text{H}} \nonumber \\
			          &\qquad \qquad \qquad \qquad \qquad \text{for  } k\in \mathcal{U}_i, i\in \mathcal{B}; \\
			\text{MSE}_{i,k} & = 1,  \quad \text{for  } k\in \mathcal{U}_i^c, i\in \mathcal{B}.
		\end{align}
	\end{subequations}
	
	For the data transmission phase, the received signal in \eqref{equ:data model} can be rewritten as 
	\begin{align}
		\mathbf{y}_{i}^\text{d} &= \sum_{k\in \mathcal{U}_i} h_{i,k} \gamma_{i,k} \mathbf{x}_{k}^\text{d} + \sum_{k\in \mathcal{U}_i^c} h_{i,k} \gamma_{i,k} \mathbf{x}_{k}^\text{d} + \mathbf{z}_{i}^\text{d}\nonumber \\
								&= \sum_{k\in \mathcal{U}_i} \hat{h}_{i,k} \gamma_{i,k} \mathbf{x}_{k}^\text{d} + \mathbf{v}_i \nonumber\\
								&= \sum_{k\in \mathcal{U}} \hat{h}_{i,k} \gamma_{i,k} \mathbf{x}_{k}^\text{d} + \mathbf{v}_i,\ i\in \mathcal{B} \label{equ:data model2}
	\end{align}
	where the last step follows from \eqref{equ:estHc}, and $\mathbf{v}_i$ represents the equivalent interference-plus-noise given as
	\begin{equation}\label{equ:IPN}
		\mathbf{v}_i = \sum_{k\in \mathcal{U}} (h_{i,k} - \hat{h}_{i,k}) \gamma_{i,k} \mathbf{x}_{k}^\text{d} + \mathbf{z}_{i}^\text{d}.
	\end{equation}
	The correlation of $\{\mathbf{x}_{k}^\text{d}\}$ is given by
	\begin{subequations}
		\begin{align}
			R_{x_{k}^\text{d}} \triangleq &\frac{1}{(1-\alpha T)} \mathbb{E}[\mathbf{x}_{k}^\text{d} (\mathbf{x}_{k}^\text{d})^\text{H}] = \beta_k^\prime P_0, \ k\in \mathcal{U},\\
			\text{and}\qquad &\frac{1}{(1-\alpha T)} \mathbb{E}[\mathbf{x}_{k}^\text{d} (\mathbf{x}_{m}^\text{d})^\text{H}] = 0, \ \forall k\neq m, \ k,m\in \mathcal{U}.
		\end{align}
	\end{subequations}
	By definition, we obtain
	\begin{subequations}
	\begin{align}
		\sigma_{v_i}^2 = &\frac{1}{(1-\alpha T)} \mathbb{E} [\mathbf{v}_i \mathbf{v}_i^\text{H}] = \sum_{k=1}^K \gamma_{i,k}^2 \beta_k^\prime P_0\cdot \text{MSE}_{i,k} + N_0, \ i\in \mathcal{B},\\
		\text{and}\quad &\frac{1}{(1-\alpha T)} \mathbb{E} [\mathbf{v}_i \mathbf{v}_j^\text{H}] = 0, \ \forall i\neq j,\ i,j \in \mathcal{B}.
	\end{align}
	\end{subequations}
	
	We next express \eqref{equ:data model2} in a matrix form. Define $\hat{\mathbf{H}}$ as the estimated channel matrix with $(i,k)$th element given by $\hat{H}_{i,k} = \hat{h}_{i,k} \gamma_{i,k}, i\in \mathcal{B}, k\in \mathcal{U}$. Denote by $\mathbf{V} = [\mathbf{v}_1^\text{T}, \mathbf{v}_2^\text{T},\cdots, \mathbf{v}_N^\text{T}]^\text{T}$, $\mathbf{Y}^\text{d} = [(\mathbf{y}_{1}^\text{d})^\text{T}, (\mathbf{y}_{2}^\text{d})^\text{T},\cdots, (\mathbf{y}_{N}^\text{d})^\text{T}]^\text{T}$, and $\mathbf{X}^\text{d} = [(\mathbf{x}_{1}^\text{d})^\text{T}, (\mathbf{x}_{2}^\text{d})^\text{T},\cdots, (\mathbf{x}_{N}^\text{d})^\text{T}]^\text{T}$. Then
	\begin{equation}\label{equ:all data model}
		\mathbf{Y}^\text{d} = \hat{\mathbf{H}} \mathbf{X}^\text{d} + \mathbf{V}.
	\end{equation}
	Note that the interference-plus-noise term $\mathbf{V}$ is in general correlated with the signal part $\hat{\mathbf{H}} \mathbf{X}^\text{d}$. Therefore, the achievable rate for \eqref{equ:all data model} is lower bounded by the case with independent Gaussian noise \cite{Medard2000}. Specifically, the throughput lower bound is given by
	\begin{equation}
		I(\mathbf{X}^\text{d};\mathbf{Y}^\text{d}|\hat{\mathbf{H}}) = \log \det \left(\mathbf{I} + \mathbf{R}_V^{-1} \hat{\mathbf{H}} \mathbf{R}_{X^\text{d}} \hat{\mathbf{H}}^\text{H} \right)
	\end{equation}
	where $\mathbf{R}_V = \text{diag}\{\sigma_{v_1}^2,\cdots,\sigma_{v_N}^2\}$ is a diagonal matrix formed by $\{\sigma_{v_n}^2\}$, $\mathbf{R}_{X^\text{d}} = \text{diag}\{R_{x_{1}^\text{d}},\cdots,R_{x_{K}^\text{d}}\}$, and $I(\mathbf{X}^\text{d};\mathbf{Y}^\text{d}|\hat{\mathbf{H}})$ is the conditional mutual information between $\mathbf{X}^\text{d}$ and $\mathbf{Y}^\text{d}$ provided that $\mathbf{x}_{k}^\text{d}$, the $k$-th row of $\mathbf{X}^\text{d}$, is independently drawn from $\mathcal{CN}(0,\beta_k^\prime P_0 \mathbf{I})$ for $k=1,\cdots, K$, and $\mathbf{v}_i$ is independently drawn from $\mathcal{CN}(0,\sigma_{v_i}^2 \mathbf{I} )$ for $i = 1,\cdots, N$. Considering the two-phase transmission scheme, we obtain the information throughput of the system:
	\begin{equation}
		R = (1-\alpha) \mathbb{E}\left[\log\det\left( \mathbf{I} + \mathbf{R}_V^{-1} \hat{\mathbf{H}}  \mathbf{R}_{X^\text{d}} \hat{\mathbf{H}}^\text{H} \right) \right].
	\end{equation}

\ifCLASSOPTIONcaptionsoff
  \newpage
\fi


%

\end{document}